\documentclass[11pt]{article}

\usepackage[utf8]{inputenc}
\usepackage[T1]{fontenc}
\usepackage{amsmath,amssymb,amsthm}
\usepackage{mathtools}
\usepackage[margin=1.1in]{geometry}
\usepackage{hyperref}
\usepackage{algorithm}
\usepackage{algpseudocode}
\usepackage{xcolor}
\usepackage{authblk}

\theoremstyle{plain}
\newtheorem{theorem}{Theorem}[section]
\newtheorem{proposition}[theorem]{Proposition}
\newtheorem{lemma}[theorem]{Lemma}
\newtheorem{corollary}[theorem]{Corollary}

\theoremstyle{definition}
\newtheorem{definition}[theorem]{Definition}

\DeclareMathOperator{\GL}{GL}
\DeclareMathOperator{\Mon}{Mon}
\DeclareMathOperator{\Aut}{Aut}

\DeclareMathOperator{\rank}{rank}
\newcommand{\F}{\mathbb{F}}
\newcommand{\Sn}{\mathcal{S}_n}
\newcommand{\C}{\mathcal{C}}
\newcommand{\K}{\mathcal{K}}
\newcommand{\diag}{\operatorname{diag}}
\algrenewcommand\algorithmicrequire{\textbf{Input:}}
\algrenewcommand\algorithmicensure{\textbf{Output:}}
\algnewcommand\algorithmicbreak{\textbf{break}}
\algnewcommand\Break{\algorithmicbreak}
\DeclareMathOperator{\rref}{rref}
\DeclareMathOperator{\piv}{piv}
\newcommand{\NO}{\textsc{no}}
\newcommand{\YES}{\textsc{yes}}

\newcommand{\new}[1]{\textcolor{black}{#1}}
\newenvironment{gce}{\begingroup\color{black}}{\endgroup}

\title{Search-to-Decision Reductions for the\\ Linear \new{and General} Code Equivalence Problems}

\author{Abhinaba Mazumder}
\affil{University of Zurich\\ Email: abhinaba.mazumder@math.uzh.ch}

\begin{document}
\maketitle

\begin{abstract}
In this paper, we present efficient search-to-decision reductions for the Linear Code Equivalence (LCE) and Generalised Code Equivalence (GCE) problems. Our methodology is inspired by the recent search-to-decision reduction for Permutation Code Equivalence (PCE) in \cite{BM23}. We demonstrate how to recover the permutation component of the equivalence using a decision oracle, and subsequently show one way of recovering the diagonal and field automorphism components in deterministic polynomial time by leveraging the elegant Engel-Schneider algorithm for diagonal equivalence.
\end{abstract}

\section{Introduction}

The Code Equivalence (CE) problem fundamentally asks whether two linear codes are identical up to a specified set of metric-preserving transformations. It is one of a class of several equivalence problems, similar philosophically to graph isomorphism, and lattice isomorphism, for instance. Depending on the allowed transformations, this question yields three main variants: Permutation Code Equivalence (PCE), Linear Code Equivalence (LCE), and Generalised Code Equivalence (GCE). The assumed computational hardness of these problems has made them a promising foundation for post-quantum cryptographic schemes, most notably the Linear Equivalence Signature Scheme (LESS) \cite{LESS20} submitted to the NIST PQC competition.

Recently, a search-to-decision reduction for PCE was introduced in \cite{BM23}. However, extending this reduction to the linear and generalized cases remained an open problem, as mentioned in \cite{BM23}, and as highlighted in a recent survey on the topic \cite{survey24}. In this paper, we resolve these open problems, showing that search-to-decision reductions for both LCE and GCE naturally follow from an adaptation of the PCE framework.

The remainder of the paper is organized as follows. Section \ref{sec:background} provides the necessary background and formal definitions. Section \ref{sec:main} details the first step of our reduction, where we recover the permutation. Section \ref{sec:diagonal} explains how to efficiently recover the diagonal scaling component using the Engel-Schneider algorithm. Section \ref{sec:general} extends these results to solve the Generalised Code Equivalence problem. Finally, Section \ref{sec:conclusion} concludes the paper and discusses future work.

In this paper, we prove the following statement.

\begin{theorem}\label{thm:main}
There is an efficient reduction from search-LCE to LCE. 
\end{theorem}

The reduction consists of two steps.

\begin{itemize}
\item Using the oracle for LCE, determine a permutation $\pi\in\Sn$ which is realised as a part of some monomial equivalence between $G$ and $F$. 
\item Given such a $\pi$, determine the diagonal matrix. (no oracle in this step)
\end{itemize}

Analogous to the multiset of repeated columns in \cite{BM23}, we use the partition of the coordinates into projective classes. Essentially ``identical columns'' is replaced everywhere by ``proportional columns'', and the exact strategy of \cite{BM23} can be carried out. As far as possible, we keep the exposition of \cite{BM23}, so that the two arguments can be read side by side. 

\begin{gce}
We then show, in Section \ref{sec:general}, that the same argument naturally settles the most general variant of the problem, in which the equivalence also involves a field automorphism.

\begin{theorem}\label{thm:general}
There is an efficient reduction from search-GCE to GCE.
\end{theorem}

As many have noted, including \cite{BM23}, we essentially get GCE for free is we can solve LCE, given than we simply have to cycle through $\operatorname{log}_p(q)$ LCE instances. For our purposes, it is fortunate that a field automorphism respects projective classes (Lemma \ref{lem:frobclass}), and so the first step of the reduction remains the same. We then simply enumerate $\Aut(\F_q)$, which is cyclic of order $\log_p q$, and do the second step for each candidate.  
\end{gce}
\\
Implementations of the reductions can be found in \url{https://git.math.uzh.ch/projects/3368}.

\section{Background}\label{sec:background}

In this section, we fix the notation we will use for the rest of the paper, as well as providing the background we need for the reductions. An $(n,k,d)$-linear code $\C$ over $\F_q$ is a $k$-dimensional $\F_q$-subspace of $\F_q^n$. A matrix $G\in\F_q^{k\times n}$ is called a generator matrix of $\C$ if its rows span $\C$. The invertible matrices of dimension $k$ over $\F_q$ are denoted by $\GL_k(\F_q)$. Given a generator matrix $G$ of $\C$ and a change of basis matrix $S\in\GL_k(\F_q)$, the matrix $SG$ also generates $\C$. We denote by $g_i$ the $i$-th column of $G$. We also abuse notation and write, for instance, $(GQ)_j$ to denote the $j-$th column of the matrix $GQ$, and $(GQ)_{i, j}$ for the entry in row $i$ and column $j$.

To each permutation $\pi\in\Sn$ we associate the permutation matrix $P\in\{0,1\}^{n\times n}$ having a $1$ in entry $(i,j)$ if and only if $\pi(i)=j$. Thus, $\pi(v)=vP$. A matrix $M\in\F_q^{n\times n}$ is {monomial} if it has exactly one nonzero entry in each row and each column; equivalently $M = PD$ with $P$ a permutation matrix and $D=\diag(d_1,\dots,d_n)$, $d_j\in\F_q^*$. We write $\Mon_n$ for the group of monomial matrices and $\pi_M$ for the permutation attached to the permutation matrix of $M$. Finally, in pseudocode and elsewhere, $A\,\|\,u^{(s)}$ denotes that we append $s$ copies of $u$ to $A$.

\begin{definition}[Matrix version of LCE]\label{def:lce}
Let $\C_1,\C_2$ be two input linear codes over $\F_q$ with generator matrices $G,F$ respectively. LCE consists in deciding whether there exist $S\in\GL_k(\F_q)$ and $M\in\Mon_n$ such that $SGM=F$.
\end{definition}

We define by search-LCE the problem of finding such a pair $(S,M)$. As in \cite{BM23}, we say that such an $M$ {realises} the equivalence between $G$ and $F$.

\begin{gce}
For GCE, we require a field automorphism. If $q=p^e$ with $p$ prime, we recall that
\[
\Aut(\F_q) \;=\; \big\langle\, x\mapsto x^{p} \,\big\rangle \;\cong\; \mathbb{Z}/e , 
\]
is cyclic of order $e$, and generated by the Frobenius. In particular, it is trivial when $q$ is prime. Let $\alpha\in\Aut(\F_q)$ be the $s$-th power of the Frobenius map. For a matrix $A\in\F_q^{k\times n}$, by $\alpha(A)$ we denote the matrix obtained by applying $\alpha$ to each entry of $A$, so that $\alpha(A)_{ij}=A_{ij}^{p^s}$. The group of semilinear isometries of $\F_q^n$ is then $\Gamma_n:=\Mon_n\rtimes\Aut(\F_q)$, and we have the following problem.

\begin{definition}[Matrix version of GCE]\label{def:gce}
Let $\C_1,\C_2$ be two input linear codes over $\F_q$ with generator matrices $G,F$ respectively. GCE consists in deciding whether there exist $S\in\GL_k(\F_q)$, $M\in\Mon_n$ and $\alpha\in\Aut(\F_q)$ such that
\[
S\,\alpha(G)\,M \;=\; F .
\]
\end{definition}

We define by search-GCE the problem of finding such a triple $(S,M,\alpha)$. Note that since $\alpha$ is a ring homomorphism, $\alpha(GM)=\alpha(G)\alpha(M)$ and $\alpha(M)$ ranges over $\Mon_n$ as $M$ does.
\end{gce}

Some useful things to remember are as follows:

\begin{lemma}\label{lem:matrix}
Let $A\in\F_q^{k\times n}$, let $P$ be the permutation matrix of $\pi\in \mathcal{S}_n$ and let $D=\diag(d_1,\dots,d_n)$. Then:
\begin{enumerate}
\item[(i)] $(APD)_{j} = d_j\, a_{\pi^{-1}(j)}$ for all $j$;
\item[(ii)] $D P = PD^{\pi}$, where $D^{\pi} := \diag(d_{\pi^{-1}(1)},\dots,d_{\pi^{-1}(n)})$;
\item[(iii)] for $M_1,M_2\in\Mon_N$ the product $M_1M_2$ is monomial and $\pi_{M_1M_2} = \pi_{M_2}\circ\pi_{M_1}$.
\end{enumerate}
\end{lemma}

\begin{proof}
(i) We have $(AP)_{a,j} = \sum_i A_{a,i}P_{i,j}$ and $P_{i,j}\neq 0$ exactly when $i=\pi^{-1}(j)$, so $(AP)_{j} = a_{\pi^{-1}(j)}$; as right multiplication by $D$ scales the $j$-th column by $d_j$.

(ii) $(D P)_{a,j} = D_a P_{a,j}$, which is nonzero if and only if $\pi(a)=j$, with value $d_a$. Also $(PD^\pi)_{a,j} = P_{a,j}d_{\pi^{-1}(j)}$, nonzero if and only if $\pi(a)=j$, in which case $d_{\pi^{-1}(j)}=d_a$.

(iii) Writing $M_i = P_iD_i$ and using (ii), $M_1M_2 = P_1D_1P_2D_2 = P_1P_2D_1^{\pi_2}D_2$, whose permutation part is $P_1P_2$, which corresponds to the permutation $\pi_2\circ \pi_1$.
\end{proof}

By Lemma \ref{lem:matrix}(i), writing $M=PD$, $\pi=\pi_M$, we have
\begin{equation}\label{eq:star}
SGM = F \quad\Longleftrightarrow\quad d_j S\,g_{\pi^{-1}(j)} \;=\; f_j \quad (1\le j\le n),
\end{equation}
equivalently $d_{\pi(i)}Sg_i = d_{\pi(i)}f_{\pi(i)}$ for all $i$. We refer to \eqref{eq:star} throughout. Also, part (ii) of Lemma \ref{lem:matrix} tells us that whether the LCE is posed as $SGPD=F$ or $SGDP=F$ is immaterial. Additionally, we keep in mind that the structure of a monomial matrix, i.e where its non-zero entries lie, is determined entirely by the permutation, and the diagonal matrix simply scales the entries.

\subsection{Projective classes}

The following partition plays, in the linear case, the role played in \cite{BM23} by the partition of the coordinates according to equality of columns.

\begin{definition}\label{def:classes}
For nonzero $u,v\in\F_q^k$ write $u\sim v$ if $v\in\F_q^* u$; this is an equivalence relation. Let $G\in\F_q^{k\times n}$ have no zero column. We denote by $\K(G)$ the partition of $\{1,\dots,n\}$ in which $i$ and $i'$ lie in the same part, which we call a {class}, if and only if $g_i\sim g_{i'}$. We set
\[
m(G) \;:=\; \max_{K\in\K(G)} |K| .
\]
\end{definition}

\begin{lemma}\label{lem:invariance}
Let $G,F\in\F_q^{k\times n}$ have no zero column and suppose $SGM=F$ with $S\in\GL_k(\F_q)$, $M\in\Mon_N$ and $\pi := \pi_M$. Then for all $i,i'$,
\[
g_i \sim g_{i'} \quad\Longleftrightarrow\quad f_{\pi(i)}\sim f_{\pi(i')} .
\]
Consequently $\pi$ maps each class of $G$ onto a class of $F$, inducing a bijection $\K(G)\to\K(F)$ which preserves cardinalities. In particular the multisets of class sizes of $G$ and $F$ coincide, and $m(G)=m(F)$.
\end{lemma}

\begin{proof}
By \eqref{eq:star}, $f_{\pi(i)} = d_{\pi(i)} S g_i$ with $d_{\pi(i)}\neq 0$. If $g_{i'}=\lambda g_i$ with $\lambda\neq 0$, then
\[
f_{\pi(i')} = d_{\pi(i')}Sg_{i'} = \lambda d_{\pi(i')}Sg_i = \lambda d_{\pi(i')}d^{-1}_{\pi(i)}\, f_{\pi(i)} ,
\]
a nonzero multiple of $f_{\pi(i)}$. By applying the same computation to $S^{-1}FM^{-1}=G$, we see that $\pi$ is a bijection preserving $\sim$ in both directions, and thus it maps classes onto classes.
\end{proof}

\begin{gce}
In the next lemma, we note that a field automorphism respects the classes of Definition \ref{def:classes}. This will allow us to seamlessly extend our reduction to GCE.

\begin{lemma}\label{lem:frobclass}
Let $G\in\F_q^{k\times n}$ have no zero column and let $\alpha\in\Aut(\F_q)$. Then $\alpha(G)$ has no zero column, $\rank\alpha(G)=\rank G$, and
\[
\K\big(\alpha(G)\big) \;=\; \K(G) \qquad\text{as partitions of } \{1,\dots,n\}, \qquad\text{so } m\big(\alpha(G)\big)=m(G).
\]
Moreover $\alpha\big(G\,\|\,u^{(s)}\big)=\alpha(G)\,\|\,\alpha(u)^{(s)}$ for every column $u$ and every $s\ge0$.
\end{lemma}

\begin{proof}
For the first claim, we note that $\alpha$ is a bijection of $\F_q$ with $\alpha(0)=0$, and thus we have $\alpha(g_i)=0$ if and only if $g_i=0$. Now, let $r=\rank G$ and pick an $r\times r$ submatrix $N$ of $G$ with $\det N\neq 0$. The corresponding submatrix of $\alpha(G)$ is $\alpha(N)$ and, because $\det$ is a polynomial with coefficients in $\F_p$ and $\alpha$ a ring homomorphism fixing $\F_p$, we get $\det\alpha(N)=\alpha(\det N)\neq 0$. Hence $\rank\alpha(G)\ge r$. Applying the same argument to $\alpha^{-1}$ gives us the equality of ranks.

Further, suppose $g_{i'}=\lambda g_i$ for some $\lambda\in\F_q^*$. Then $\alpha(g_{i'})=\alpha(\lambda)\,\alpha(g_i)$ with $\alpha(\lambda)\in\F_q^*$, so $\alpha(g_i)\sim\alpha(g_{i'})$. Similarly as above, the other direction follows by looking at $\alpha^{-1}$. Thus $g_i\sim g_{i'}$ if and only if $\alpha(g_i)\sim\alpha(g_{i'})$, and the claims about the partitions and partition sizes follow. The last claim is immediate, with $\alpha$ being applied entrywise.
\end{proof}
\end{gce}

\subsection{Preprocessing}\label{sec:prep}

We may assume the following about the input, at the cost of $O(nk^2)$ field operations and no oracle calls.

\begin{lemma}\label{lem:prep}
Let $G,F\in\F_q^{k\times n}$.
\begin{enumerate}
\item[(i)] If $\rank(G)\neq\rank(F)$ then $(G,F)$ is a NO instance of LCE. Otherwise both may be replaced by full row rank generator matrices of the same codes, at the cost of a known change of basis.
\item[(ii)] Let $Z_G := \{i : g_i=0\}$ and $Z_F := \{j : f_j = 0\}$, and let $\bar G,\bar F$ be obtained by deleting the corresponding columns. Then $(G,F)$ is a YES instance if and only if $|Z_G|=|Z_F|$ and $(\bar G,\bar F)$ is a YES instance; moreover a solution for $(\bar G,\bar F)$ extends to one for $(G,F)$ by any bijection $Z_G\to Z_F$ with all scalars equal to $1$.
\end{enumerate}
\end{lemma}

\begin{proof}
(i) is immediate.

(ii) Suppose $SGM=F$ with $\pi=\pi_M$. By \eqref{eq:star}, $Sg_{\pi^{-1}(j)} = d^{-1}_jf_j$ with $S$ invertible and $d^{-1}_j\neq 0$, so $g_{\pi^{-1}(j)}=0$ if and only if $f_j=0$. Hence $\pi$ maps $Z_G$ bijectively onto $Z_F$ and its complement onto the complement of $Z_F$; restricting $M$ to the latter block gives an equivalence of $\bar G,\bar F$. Conversely, given $\bar S\bar G\bar M = \bar F$, extend $\bar M$ as described; relation \eqref{eq:star} holds trivially on the deleted coordinates, both sides being zero. Finally $\rank\bar G = \rank G$, since deleting zero columns does not change the column rank.
\end{proof}

\new{As we saw earlier, ranks and zero columns are unaffected by $\alpha$. Thus, both parts of Lemma \ref{lem:prep} hold verbatim for GCE}

{From now on $G$ and $F$ are assumed to have full row rank $k$ and no zero column}, and $n$ denotes their common length. We also assume, having spent one oracle call, that $(G,F)$ is a YES instance of LCE. By Lemma \ref{lem:invariance} we may then set
\[
m := m(G) = m(F).
\]

\section{Recovering the permutation}\label{sec:main}

This is the first step of our reduction. This section is also common to both LCE and GCE reductions, as we justify later, but for now, we stick to the LCE formulation. Here, given access to an oracle, we find index-by-index the permutation associated with the monomial. We search iteratively for the image $\pi(t)$ for $t=1,2,\dots,n$, where $\pi = \pi_M$ for a monomial matrix $M$ such that $SGM=F$ for the matrices $G,F$ generating the two input codes. We ask if $\pi(t)=i_t$ by extending the generator matrix $G$ with repetitions of its $t$-th column, and by extending the generator matrix $F$ with repetitions of its $i_t$-th column. As we will show, the LCE oracle only returns a positive answer if $\pi(t)=i_t$, and we can set $i_t$ as the image of $\pi$ at $t$. Further extending $G, F$ with calls to the LCE oracle leads us to the permutation. As in \cite{BM23}, we cannot simply compute the image of $\pi$ at all $t$ with independent calls to the oracle, because there might exist multiple solutions to LCE that do not agree with each other at all indices. 

The difference with \cite{BM23} is that the appended columns are now "pinned" only up to a scalar, because rescaling an appended coordinate produces a linearly equivalent instance and therefore cannot change the answer of the oracle. Hence, we work with the classes of Definition \ref{def:classes} throughout. The oracle can be made to certify a matching of projective classes, and Lemma \ref{lem:normalise} converts such a matching into a statement about indices.

We start with a lemma that allows us to assume a certain choice for the images of a solution $M$ at indices lying in a common projective class. We show that if the columns $g_t,g_{t_1},\dots,g_{t_k}$ are all proportional, so are the columns $f_{\pi(t)},f_{\pi(t_1)},\dots,f_{\pi(t_k)}$. Hence a monomial matrix $M'$ whose permutation agrees with $\pi_M$ on inputs other than $t,t_1,\dots,t_k$ will also satisfy $SAM'=B$, for a suitable choice of the diagonal part.

\begin{lemma}\label{lem:normalise}
Let $G,F\in\F_q^{k\times n}$ have no zero columns and suppose $SGM=F$ for some $S\in\GL_k(\F_q)$ and $M\in\Mon_n$; write $\pi:=\pi_M$. Let $\mathcal{T} \subseteq \K(G)$ be any collection of classes and, for each $K\in\mathcal{T}$, let $\sigma_K : K \to \pi(K)$ be an arbitrary bijection. Then there exists $M'\in\Mon_n$ such that
\[
SGM'=F, \qquad \pi_{M'}\big|_{K}=\sigma_K \ \ (K\in\mathcal{T}), \qquad \pi_{M'}\big|_{\,\{1,\dots,n\}\setminus\bigcup\mathcal{T}} = \pi\big|_{\,\{1,\dots,n\}\setminus\bigcup\mathcal{T}} .
\]
\end{lemma}

\begin{proof}
Define $\rho:\{1,\dots,n\}\to\{1,\dots,n\}$ by $\rho|_K := \pi^{-1}\circ \sigma_K$ for $K\in\mathcal{T}$ and $\rho(i)=i$ for all $i \notin \bigcup \mathcal{T}$. Each $\rho|_K$ is a bijection from $K$ onto $K$, since $\sigma_K(K)=\pi(K)$ and $\pi^{-1}(\pi(K))=K$. By Lemma \ref{lem:invariance}, the classes in $\mathcal{T}$ are pairwise disjoint, so $\rho\in\mathcal{S}_n$. Let $Q$ be its permutation matrix.

We first observe that any permutation of the columns of $G$ within projective classes leads to a {monomial} automorphism of $G$. Indeed, by Lemma \ref{lem:matrix}(i), $(GQ)_{j}=g_{\rho^{-1}(j)}$. If $j\notin\bigcup \mathcal{T}$ then $\rho^{-1}(j)=j$ and we set $\delta_j:=1$; if $j\in K\in \mathcal{T}$ then $\rho^{-1}(j)\in K$, so $g_{\rho^{-1}(j)}\sim g_j$, i.e.\ $g_{\rho^{-1}(j)} = \delta_j g_j$ for a unique nonzero $\delta_j\in\F_q^*$, because $g_j\neq 0$. Hence
\begin{equation}\label{eq:auto}
GQ = G\Delta, \qquad \Delta := \diag(\delta_1,\dots,\delta_n)\in\GL_n(\F_q).
\end{equation}

Now write $M=PD$. Using \eqref{eq:auto}, then Lemma \ref{lem:matrix}(ii), then the commutativity of diagonal matrices,
\[
S G Q M \;=\; S G \Delta P D \;=\; S G P \Delta^{\pi} D \;=\; S G P D \Delta^{\pi} \;=\; F \Delta^{\pi} .
\]
Set $M' := Q\,M\,(\Delta^{\pi})^{-1}$, a product of monomial matrices and hence monomial. Then $SGM'=F$. The permutation part is the matrix $QP$, which corresponds to the permutation $\pi\circ\rho$. For $j\in K\in \mathcal{T}$ this equals $\pi(\pi^{-1}(\sigma_K(j)))=\sigma_K(j)$, and for $j\notin\bigcup\mathcal{T}$, it equals $\pi(j)$. Note that $S$ remains unchanged.
\end{proof}


\new{We note here that Lemma \ref{lem:normalise} as a statement about linear equivalence can be applied, in Section \ref{sec:general}, to the pair $\big(\alpha(G),F\big)$ for a fixed $\alpha$. We do not need a version of it for semilinear isometries, and indeed none is readily available, since exchanging two columns related by a nontrivial $\alpha$ cannot be achieved by scaling columns.}

\subsection{Base case of the induction}

Here we prove the base case of our induction, i.e, given $G,F$ linearly equivalent, we show how to compute the image of one element of $\{1,\dots,n\}$ via $\pi$ if we have access to an oracle. Without loss of generality, we assume that we need to compute the image of $1$. We run a loop over all possible candidates. This means that, given $i_1\in\{1,\dots,n\}$, we want to decide whether there is a monomial matrix $M$ such that $SGM=F$ for some $S\in\GL_k(\F_q)$ and $\pi_M(1)=i_1$. We stop when we find the first $i_1$ that satisfies this requirement. Recall $m=m(G)=m(F)$. Let us append $m$ copies of $g_1$ to $G$, thus constructing the matrix
\[
G_1 = \big(G \,\big|\, \underbrace{g_1,\dots,g_1}_{m}\big),
\]
and let us append $m$ copies of $f_{i_1}$to $F$, constructing the matrix
\[
F_1 = \big(F \,\big|\, \underbrace{f_{i_1},\dots,f_{i_1}}_{m}\big).
\]
We then ask the oracle if the code generated by $G_1$ is linearly equivalent to the code generated by $F_1$.

\begin{proposition}\label{prop:base}
If the answer is NO, then $\pi_M(1)\neq i_1$ for any monomial matrix $M$ realising the equivalence between $G$ and $F$. If the answer is YES, there are $S\in\GL_k(\F_q)$ and a monomial matrix $M$ with $SGM=F$ and $\pi_M(1)=i_1$.

\end{proposition}

\begin{proof}
For the first claim, we suppose the contrary, i.e.\ the answer is negative and there exists $M=PD$ realising the equivalence between $G,F$ with $\pi:=\pi_M$ satisfying $\pi(1)=i_1$. Define $\pi'\in\mathcal{S}_{n+m}$ to be equal to $\pi$ on $\{1,\dots,n\}$ and such that $\pi'(n+s)=n+s$ for all $s\in\{1,\dots,m\}$, and let $M'=P'D'$ with $d'_j := d_j$ for $j\le n$ and $d'_{n+s} := d_{i_1}$ for $s\in\{1,\dots,m\}$. For $j\le n$ we have $(\pi')^{-1}(j)=\pi^{-1}(j)\le n$, so the $j$-th column of $SG_1M'$ is $d_j S g_{\pi^{-1}(j)} = f_j$. For $j = n+s$ we get $d'_{n+s}\,Sg_1 = d_{i_1}d^{-1}_{\pi(1)}f_{\pi(1)} = f_{i_1}$. Hence $SG_1M' = F_1$, so there is at least one equivalence, and therefore the answer could not have been negative.

Suppose now that the answer is positive, then we have to prove that $M$ can be chosen such that $\pi_M(1)=i_1$. Let $c$ be the size of the class of $g_1$ in $\K(G)$ and $c'$ the size of the class of $f_{i_1}$ in $\K(F)$. Note that $1\le c,c'\le m$. In $G_1$ the class of $g_1$ has size $c+m$, and every other class has size at most $m$; so $G_1$ has exactly one class of size $>m$. The same holds for $F_1$, with the class of $f_{i_1}$ having size $c'+m$.

Let $M'$ be a monomial matrix realising the equivalence of $G_1$ and $F_1$ and $S$ an invertible matrix such that $SG_1M'=F_1$. By Lemma \ref{lem:invariance}, $\pi_{M'}$ induces a cardinality-preserving bijection $\K(G_1)\to\K(F_1)$. Since each side has exactly one class of size $>m$, these two classes must correspond, forcing $c=c'$ and
\[
\pi_{M'}\big(\{\text{class of } g_1 \text{ in } G_1\}\big) = \{\text{class of } f_{i_1} \text{ in } F_1\} .
\]
Using Lemma \ref{lem:normalise} applied to this class, $M'$ can be chosen such that $\pi_{M'}(1)=i_1$, such that $\pi_{M'}(n+s)=n+s$ for all $s\in\{1,\dots,m\}$, and such that the remaining $c-1$ indices of the class of $g_1$ inside $\{1,\dots,n\}$ are sent onto the remaining $c'-1=c-1$ indices of the class of $f_{i_1}$ inside $\{1,\dots,n\}$. Moreover, every other class of $G_1$ is contained in $\{1,\dots,n\}$ and is sent by $\pi_{M'}$ to a class of $F_1$ of size at most $m$ contained in $\{1,\dots,n\}$. Thus, $\pi_{M'}$ maps $\{1,\dots,n\}$ onto itself and satisfies $\pi_{M'}(1)=i_1$.

It follows that $M'$ can be treated as a block diagonal with respect to the splitting $\{1,\dots,n\}\cup\{n+1,\dots,n+m\}$. Let $M$ be its top left $n\times n$ block, a monomial matrix with $\pi_M = \pi_{M'}|_{\{1,\dots,n\}}$. Comparing the $j$-th columns of $SG_1M'=F_1$ for $j\le n$, and using $(\pi_{M'})^{-1}(j)\le n$, gives $SGM = F$. In particular $M$ realises an equivalence between $G$ and $F$ with $\pi_M(1)=i_1$.
\end{proof}

\subsection{The general step}

Now, we assume that we found a set $A_{T-1}$ of $T-1$ indices whose images by a monomial matrix $M_{T-1}$ that realises the equivalence between $G$ and $F$ are known. That is, we know that there exist $S_{T-1}\in\GL_k(\F_q)$ and a monomial matrix $M_{T-1}$ such that
\[
S_{T-1}GM_{T-1}=F, \qquad \pi_{M_{T-1}}(t)=i_t \ \text{ for all } t\in A_{T-1}, \qquad |A_{T-1}|=T-1 .
\]
As in \cite{BM23}, we should notice that it is important to carry over the information of the previous steps, as we need to be able to extend the ``same'' equivalence step by step. When we move to the next element in the domain, we need to avoid finding an image of another equivalence, for example in the case in which the monomial automorphism group of the code acts transitively, thus allowing any possible image for the index $t$.

Let $t_1,\dots,t_\ell$ be a list of elements of $A_{T-1}$ containing exactly one element from each class of $\K(G)$ meeting $A_{T-1}$. We call them the {representatives}, and we write $K_j\in\K(G)$ for the class of $g_{t_j}$ and $c_j := |K_j|$. Notice that we must have $\ell\le |A_{T-1}|=T-1$. Set $L_j := \pi_{M_{T-1}}(K_j) \in \K(F)$. By Lemma \ref{lem:invariance}, $L_j$ is the class of $f_{\pi_{M_{T-1}}(t_j)}$ in $F$, and $|L_j|=c_j$. Note that $L_j$ does not depend on the choice of $M_{T-1}$. The classes $L_1,\dots,L_\ell$ are pairwise distinct, again by Lemma \ref{lem:invariance}, because the $K_j$s are.

We record the following property, which will let us prescribe the images of {all} committed indices, not merely of the representatives.

\begin{lemma}\label{lem:consistency}
For every $j\le\ell$ we have $\pi_{M_{T-1}}(A_{T-1})\cap L_j = \pi_{M_{T-1}}(A_{T-1}\cap K_j)$.
\end{lemma}

\begin{proof}
If $a\in A_{T-1}\cap K_j$ then $g_a\sim g_{t_j}$, so by Lemma \ref{lem:invariance} $f_{\pi_{M_{T-1}}(a)}\sim f_{\pi_{M_{T-1}}(t_j)}$, i.e.\ $\pi_{M_{T-1}}(a)\in L_j$. Conversely if $a\in A_{T-1}$ is such that $\pi_{M_{T-1}}(a)\in L_j$ then we have $g_a\sim g_{t_j}$, i.e.\ $a\in K_j$.
\end{proof}

Next, we build a query to submit to the oracle. We construct the matrix $G'_T$ by appending to $G$ $m$ columns equal to $g_{t_1}$, $2m$ columns equal to $g_{t_2}$, \dots, $\ell m$ columns equal to $g_{t_\ell}$. Similarly, we construct $F'_T$ by appending to $F$ $m$ columns equal to $f_{\pi_{M_{T-1}}(t_1)}$, $2m$ columns equal to $f_{\pi_{M_{T-1}}(t_2)}$, \dots, $\ell m$ columns equal to $f_{\pi_{M_{T-1}}(t_\ell)}$. We write $B_j$ for the block of $jm$ appended positions carrying $g_{t_j}$; the corresponding appended positions of $F'_T$ occupy the same indices, and we denote them $B'_j$.

\begin{proposition}\label{fact:equiv}
The matrices $G'_T$ and $F'_T$ span linearly equivalent codes.
\end{proposition}

\begin{proof}
Since $\exists S_{T-1}$ with $S_{T-1}GM_{T-1}=F$, write $M_{T-1}=P_{T-1} D_{T-1}$ and let $d^{-1}_j$ be as in \eqref{eq:star}, so that $S_{T-1}g_{t_j} = d^{-1}_{\pi_{M_{T-1}}(t_j)}f_{\pi_{M_{T-1}}(t_j)}$. Extend $\pi_{M_{T-1}}$ to $\pi'_T$ by $\pi'_T(i)=\pi_{M_{T-1}}(i)$ for $i\le n$ and $\pi'_T(i)=i$ for $i>n$, and extend the diagonal part by setting the entry at each position of $B_j$ equal to $d_{\pi_{M_{T-1}}(t_j)}$. Calling $M'_T$ the resulting monomial matrix and using Lemma \ref{lem:matrix}(i), one checks column by column that $S_{T-1}G'_TM'_T=F'_T$.
\end{proof}

\begin{proposition}\label{fact:agree}
There is a monomial matrix $M'_T$ and an invertible $k\times k$ matrix $S'_T$ such that $S'_TG'_TM'_T=F'_T$ and $\pi_{M_T}\big|_{A_{T-1}}=\pi_{M'_{T-1}}\big|_{A_{T-1}}$.
\end{proposition}

\begin{proof}
In $G'_T$ the class of $g_{t_j}$ is $\widehat K_j := K_j\cup B_j$, of size $c_j+jm$. Indeed the columns $g_{t_1},\dots,g_{t_\ell}$ lie in pairwise distinct classes of $G$, so the block $B_j$ joins the class of $g_{t_j}$ and no other. Similarly the class of $f_{\pi_{M_{T-1}}(t_j)}$ in $F'_T$ is $\widehat L_j := L_j \cup B'_j$, of size $c_j + jm$. Every other class of $G'_T$ or $F'_T$ has size at most $m$. Since $1\le c_j\le m$ we have, for $1\le j<\ell$,
\begin{equation}\label{eq:layers}
c_j+jm \;\le\; m+jm \;=\; (j+1)m \;<\; (j+1)m + c_{j+1},
\end{equation}
so the class sizes $c_1+m<c_2+2m<\dots<c_\ell+\ell m$ are strictly increasing, therefore unique, and all exceed $m$. Hence the classes of size $>m$ of $G'_T$ are exactly the $\widehat K_j$, those of $F'_T$ are exactly the $\widehat L_j$, and in each case their sizes are pairwise distinct. Any permutation $\pi$ as part of a monomial realising an equivalence of $G'_T$ and $F'_T$ preserves class sizes by Lemma \ref{lem:invariance}, and therefore must satisfy $\pi(\widehat K_j) = \widehat L_j$ for every $j$. Such a $\pi$ exists by Proposition \ref{fact:equiv}.

We may now apply Lemma \ref{lem:normalise} with $\mathcal{T} = \{\widehat K_1,\dots,\widehat K_\ell\}$ and, for each $j$, a bijection $\sigma_j : \widehat K_j\to\widehat L_j$ chosen to be the identity on $B_j$ (mapping $B_j$ onto $B'_j$, which occupies the same positions) and to agree with $\pi_{M_{T-1}}$ on $A_{T-1}\cap K_j$. Such a $\sigma_j$ exists, as by Lemma \ref{lem:consistency}, $\pi_{M_{T-1}}$ maps $A_{T-1}\cap K_j$ injectively into $L_j$, and $|K_j|=|L_j|=c_j$, so the partial injection extends to a bijection $K_j\to L_j$. Since $A_{T-1}\subseteq\bigcup_{j\le\ell}K_j$, the resulting $\pi_{M'_T}$ agrees with $\pi_{M_{T-1}}$ on all of $A_{T-1}$.
\end{proof}

Now, let us take an element $t\notin A_{T-1}$ and consider the column $g_t$.

\medskip
Suppose $g_t\sim g_{t_{j}}$ for some $j\le\ell$, i.e.\ $t\in K_{j}$. Find an index $i \in L_{j}\setminus\pi_{M_{T-1}}(A_{T-1})$. Such an $i$ must exist, as by Lemma \ref{lem:consistency}, $\pi_{M_{T-1}}(A_{T-1})\cap L_{j} = \pi_{M_{T-1}}(A_{T-1}\cap K_{j})$, which has cardinality $|A_{T-1}\cap K_{j}|\le c_{j}-1$ since $t\in K_{j}\setminus A_{T-1}$, while $|L_{j}|=c_{j}$. Using Lemma \ref{lem:normalise} with $\mathcal{T}=\{K_{j}\}$ and a bijection $\sigma: K_{j}\to L_{j}$ agreeing with $\pi_{M_{T-1}}$ on $A_{T-1}\cap K_{j}$ and sending $t\mapsto i$, there exists $M_T$ realising the equivalence between $G$ and $F$ with $\pi_{M_T}\big|_{A_{T-1}}=\pi_{M_{T-1}}\big|_{A_{T-1}}$ and $\pi_{M_T}(t)=i$. Then extend $A_{T-1}$ as $A_T\leftarrow A_{T-1}\cup\{t\}$

\medskip
Let us now suppose that the class of $g_t$ is not represented. We append to $G'_T$ a number $(\ell+1)m$ of columns equal to $g_t$, obtaining $G''_T$, and to $F'_T$ we append $(\ell+1)m$ columns equal to $f_{i_t}$, obtaining $F''_T$. We write $B$ and $B'$ for the two new blocks of appended positions. Now we ask the oracle if the two codes generated by $G''_T$ and $F''_T$ are linearly equivalent. This construction is represented in Algorithm \ref{alg:BQ}.

\begin{algorithm}[htbp]
\caption{\textsc{BuildQuery}}
\label{alg:BQ}
\begin{algorithmic}[1]
\Require $G,F\in\F_q^{k\times n}$; the committed set $A_{T-1}$ with its values $(i_t)_{t\in A_{T-1}}$; the list of representatives $(t_1,\dots,t_\ell)$; $m=m(G)$; an index $t\in[n]\setminus A_{T-1}$; a candidate $i\in[n]\setminus\pi(A_{T-1})$.
\Ensure the pair $(G''_T,F''_T)$ constructed before Proposition \ref{prop:step}.
\Statex
\State $G'_T\gets G$, \quad $F'_T\gets F$
\For{$j=1,\dots,\ell$}
  \State $G'_T\gets G'_T\,\|\,g_{t_j}^{(jm)}$, \quad $F'_T\gets F'_T\,\|\,f_{i_{t_j}}^{(jm)}$
\EndFor
\State $G''_T\gets G'_T\,\|\,g_{t}^{((\ell+1)m)}$, \quad $F''_T\gets F'_T\,\|\,f_{i}^{((\ell+1)m)}$
\State \Return $(G''_T,F''_T)$
\end{algorithmic}
\end{algorithm}

\begin{proposition}\label{prop:step}

If the answer is NO, then $\pi_{M_T}(t)\neq i_t$ for any monomial matrix $M_T$ realising the equivalence between $G$ and $F$ and agreeing with $\pi_{M_{T-1}}$ on all inputs from $A_{T-1}$.
If the answer is YES, there are $S_T\in\GL_k(\F_q)$ and a monomial matrix $M_T$ with $S_TGM_T=F$, $\pi_{M_T}\big|_{A_{T-1}}=\pi_{M_{T-1}}\big|_{A_{T-1}}$ and $\pi_{M_T}(t)=i_t$.

\end{proposition}

\begin{proof}
For the first claim, we again suppose the contrary. If $\pi_{M_{T-1}}$ could be extended by an $M_T$ with $\pi_{M_T}(t)=i_t$, write $M_T=P_TD_T$. Construct $M''_T$ by letting its permutation be $\pi_{M_T}$ on $\{1,\dots,n\}$ and the identity on the appended positions, and its diagonal entries be those of $D_T$ on $\{1,\dots,n\}$, equal to $d_{\pi_{M_{T-1}}(t_j)}$ on $B_j$, and equal to $d_{i_t}$ on $B$. As in the proof of Proposition \ref{fact:equiv} one checks column by column that $S_TG''_TM''_T=F''_T$. So the answer of the oracle could not have been negative.

If the answer is positive, then we have to prove that the partially identified permutation we have previously generated can be extended by setting $\pi_{M_T}(t)=i_t$. Let $M''$ be a monomial matrix realising the equivalence, so that $S''G''_TM'' = F''_T$ for some invertible $S''$, and let $\pi'' := \pi_{M''}$. From $\pi''$, we need to arrive at a permutation that not only realises the equivalence between $G$ and $F$, but also agrees with $\pi_{M_{T-1}}$ on $A_{T-1}$, maps $t$ to $i_t$, and maps $\{1,\dots,n\}$ to $\{1,\dots,n\}$.

By hypothesis the class $K$ of $g_t$ in $\K(G)$ is distinct from $K_1,\dots,K_\ell$, so, writing $c := |K|$, the class of $g_t$ in $G''_T$ is $\widehat K := K\cup B$, of size $c+(\ell+1)m$. Together with \eqref{eq:layers} and $c\ge 1$ we obtain
\[
c_1+m \;<\; c_2+2m \;<\;\cdots\;<\; c_\ell+\ell m \;\le\; (\ell+1)m \;<\; c+(\ell+1)m ,
\]
so that $G''_T$ has exactly $\ell+1$ classes of size $>m$, namely $\widehat K_1,\dots,\widehat K_\ell,\widehat K$, and their sizes are pairwise distinct.

We claim that $f_{i_t}\not\sim f_{\pi_{M_{T-1}}(t_j)}$ for every $j\le\ell$. Suppose the contrary, say $f_{i_t}\sim f_{\pi_{M_{T-1}}(t_{i})}$. Then $L_{i}$, $B'_{i}$ and $B'$ all lie in a single class of $F''_T$, and the classes of $F''_T$ of size $>m$ are the $\widehat L_j$ for $j\neq i$ together with that merged class, that is, $\ell$ classes in total; every remaining class of $F''_T$ is a class of $F$ and has size at most $m$. This contradicts the equality of the multisets of class sizes of $G''_T$ and $F''_T$ provided by Lemma \ref{lem:invariance}, since $G''_T$ has $\ell+1$ classes of size $>m$. Hence, writing $L$ for the class of $f_{i_t}$ in $\K(F)$ and $c' := |L|$, the classes of $F''_T$ of size $>m$ are exactly $\widehat L_1,\dots,\widehat L_\ell$ and $\widehat L := L\cup B'$, of size $c'+(\ell+1)m$.

By looking at the unique size of these classes, and since $\pi''$ preserves class sizes, we necessarily have $c=c'$ and
\[
\pi''(\widehat K_j) = \widehat L_j \ \ (1\le j\le\ell), \qquad \pi''(\widehat K) = \widehat L .
\]

We may therefore apply Lemma \ref{lem:normalise} to $G''_T,F''_T, S'',M''$ with $\mathcal{T}=\{\widehat K_1,\dots,\widehat K_\ell,\widehat K\}$ and the following bijections. For $j\le\ell$, take $\sigma_j : \widehat K_j\to\widehat L_j$ equal to the identity on $B_j$, agreeing with $\pi_{M_{T-1}}$ on $A_{T-1}\cap K_j$, and arbitrary from the remaining elements of $K_j$ onto the remaining elements of $L_j$. This is possible by Lemma \ref{lem:consistency} together with $|K_j|=|L_j|$. Take $\sigma:\widehat K\to\widehat L$ equal to the identity on $B$, sending $t\mapsto i_t$, and arbitrary from $K\setminus\{t\}$ onto $L\setminus\{i_t\}$; this is possible since $|K|=c=c'=|L|$, and no constraint coming from $\pi_{M_{T-1}}$ arises because $K$ contains no element of $A_{T-1}$.

Let $M''$ be the monomial matrix produced by Lemma \ref{lem:normalise} and $\pi'' := \pi_{M''}$. We have to check that $\pi''$ maps $\{1,\dots,n\}$ onto itself. Every appended position of $G''_T$ lies in one of the classes of $\mathcal{T}$, and is fixed by $\pi''$ by construction. Every index of $K_j$ is sent into $L_j\subseteq\{1,\dots,n\}$, and every index of $K$ into $L\subseteq\{1,\dots,n\}$. Finally, any other class $K'$ of $G''_T$ is contained in $\{1,\dots,n\}$ and has size at most $m$, so $\pi''(K')$ is a class of $F''_T$ of size at most $m$, hence is none of $\widehat L_1,\dots,\widehat L_\ell,\widehat L$. Since all appended positions of $F''_T$ lie in those classes, $\pi''(K')\subseteq\{1,\dots,n\}$. Thus $\pi''(\{1,\dots,n\})\subseteq\{1,\dots,n\}$, and equality follows by cardinality.

Consequently $M''$ is block diagonal with respect to the splitting of the coordinates into $\{1,\dots,n\}$ and the appended positions. Let $M_T$ be its top left $n\times n$ block, so that $\pi_{M_T}=\pi''|_{\{1,\dots,n\}}$. Comparing the $j$-th columns of $S''G''_TM''=F''_T$ for $j\le n$, and using $(\pi'')^{-1}(j)\le n$, gives $S''GM_T=F$. By construction $\pi_{M_T}(t)=\sigma(t)=i_t$ and, for $a\in A_{T-1}$, $a$ lies in some $K_j$ and $\pi_{M_T}(a)=\sigma_j(a)=\pi_{M_{T-1}}(a)$. Setting $S_T := S''$ concludes the proof.
\end{proof}

The procedure is summarised in ALgorithm \ref{alg:perm}.

\begin{algorithm}[htbp]
\caption{\textsc{RecoverPermutation}}
\label{alg:perm}
\begin{algorithmic}[1]
\Require $G,F\in\F_q^{k\times n}$ of full row rank with no zero column, such that $(G,F)$ is a \YES{} instance of LCE; an oracle $\mathcal{O}_\mathrm{LCE}$ for LCE.
\Ensure $\pi\in\Sn$ such that $SGM=F$ for some $S\in\GL_k(\F_q)$ and some $M\in\Mon_n$ with $\pi_M=\pi$.
\Statex
\State compute $\K(G)$ and $\K(F)$; \quad $m\gets m(G)$ \Comment{computed once}
\State $A_0\gets\emptyset$, \quad $\ell\gets 0$
\For{$T=1,\dots,n$}
  \State $t\gets\min\big([n]\setminus A_{T-1}\big)$
  \If{$g_t\sim g_{t_j}$ for some $j\le\ell$}
    \State $L_j\gets$ the class of $f_{i_{t_j}}$ in $\K(F)$
    \State $i_t\gets\min\big(L_j\setminus\pi(A_{T-1})\big)$
  \Else
    \For{$i\in[n]\setminus\pi(A_{T-1})$ in increasing order}
      \State $(G''_T,F''_T)\gets\Call{BuildQuery}{G,F,A_{T-1}, (t_1,\dots,t_\ell),m,t,i}$
      \If{$\mathcal{O}_\mathrm{LCE}\big(G''_T,F''_T\big)=\YES$}
        \State $i_t\gets i$, \quad $\ell\gets\ell+1$, \quad $t_\ell\gets t$
        \State \Break
      \EndIf
    \EndFor
  \EndIf
  \State $A_T\gets A_{T-1}\cup\{t\}$
\EndFor
\State \Return the permutation $\pi:t\mapsto i_t$
\end{algorithmic}
\end{algorithm}

\subsection{Cost}

To conclude the recovery of the permutation, we simply iterate these steps until $A_T=\{1,\dots,n\}$. At each step the index $t$ is chosen to be the smallest element not yet in $A_{T-1}$. If the class of $g_t$ is already represented, it is committed without any oracle calls, and otherwise the candidates $i_t$ are tried in increasing order among the indices not already accounted for, the first positive answer being retained. At the end, we have identified one of the potential permutations realising an equivalence between $G$ and $F$.

\begin{corollary}\label{cor:cost}
The above procedure terminates and performs at most $n^2$ calls to an oracle to LCE with input codes of dimension $k$ and length at most
\[
n + \frac{2}{27}n^3 + O(n^2)
\]
over $\F_q$, and returns a permutation $\pi\in\Sn$ such that there exist an invertible matrix $S$ and a monomial matrix $M$ with $\pi_M=\pi$ and $SGM=F$.
\end{corollary}

\begin{proof}
Correctness has already been discussed. The second case, where the column's class is not represented, occurs once for each class of $\K(G)$, hence at most $c := |\K(G)|\le n$ times, and each occurrence tries at most $n-|A_{T-1}|\le n$ candidates; The other case costs nothing in terms of oracle calls. This gives at most $cn\le n^2$ oracle calls.

A query built at a step with $\ell$ representatives has $n+m\big(1+2+\dots+\ell\big)+(\ell+1)m = n+m(\ell+1)(\ell+2)/2$ columns. Since $\ell+1\le c$, this is at most $n+m\,c(c+1)/2$. The classes partition $\{1,\dots,n\}$ into $c$ nonempty parts, one of which has $m$ elements, so $m\le n-c+1$; maximising $(n-c+1)c(c+1)/2$ over $c$ gives the stated bound, the maximum being attained near $c=2n/3$. All queries have $k$ rows, and appending copies of existing columns to a matrix of full row rank preserves the rank.
\end{proof}

\section{Recovery of the diagonal component}
\label{sec:diagonal}

Throughout this section $G,F\in\F_q^{k\times n}$ are generator matrices of full row rank $k$ and without zero columns, and we assume that they are linearly equivalent, i.e.\ that
\[
SGM=F \qquad\text{for some } S\in\GL_k(\F_q) \text{ and some monomial } M=PD,
\]
with $P$ a permutation matrix and $D$ an invertible diagonal matrix. We assume further that $P$ is known; it is the output of the reduction of the previous section. We show that $D$, and then $S$, can be recovered by deterministic linear algebra, without any further oracle access and in time polynomial in $n$, $k$ and $\log q$.

That the diagonal part of a linear equivalence carries no computational hardness of its own is by now a known fact, and several routes to it are available in the literature.

The first route is through canonical forms for the action of the group of invertible diagonal matrices. In \cite{CPS25}, the authors construct such canonical forms and use them to shorten LESS signatures. In \cite{DMS26}, the authors place the construction in a general framework, showing that for a group $\mathcal{G}=\mathcal{G}_1\rtimes \mathcal{G}_2$ acting on a set $X$, a polynomial-time canonical form for the classes of $X/\mathcal{G}_1$ makes the inversion problem for the induced action of $\mathcal{G}_2$ on $X/\mathcal{G}_1$ equivalent to the original inversion problem. Taking $\mathcal{G}_1$ to be the diagonal group and $\mathcal{G}_2=\Sn$ gives exactly the statement that recovering $P$ suffices. A second route is computational and direct. In \cite{BMS25} the authors note that, once the permutation is known, the scalars are obtained by solving an overdetermined linear system. Alecci and D'Alconzo \cite{AD26} record the reduction of the recovery of $M=DP$ to the recovery of $P$ alone as a known fact and refer to the canonical forms above.

Both routes carry hypotheses that we would rather not impose. The canonical forms of \cite{CPS25} are not defined for every input, so further arguments are needed.

We take a simpler route. We observe that, after passing to systematic form with respect to an information set, the recovery of $D$ becomes verbatim an instance of the {diagonal equivalence} problem for rectangular matrices, solved in complete generality over an arbitrary field in \cite{EngSch}. Their algorithm is deterministic, produces a nonsingular solution by construction, and characterises the entire solution set. The idea of working with the action of the monomial map on an information set is taken from the IS-LEP formulation of \cite{PS23}, where it is used to compress the responses of the LESS identification scheme.

\subsection{Reduction to a diagonal equivalence problem}

Set $A:=GP$, so that
\begin{equation}\label{eq:AD}
F=SGPD=SAD .
\end{equation}
Since $P$ is known, $A$ is too. Note that $A$ has full row rank $k$ and no zero column, both properties being invariant under a permutation of columns.

Let $\mathcal{I}\subseteq[n]$ with $|\mathcal{I}|=k$ be an information set for $A$, that is, a set of indices for which the submatrix $A_{\mathcal{I}}$ of the corresponding columns is invertible; such an $\mathcal{I}$ exists and is found by Gaussian elimination. In this section, for any $\mathcal{K}\subseteq[n]$ we write $A_{\mathcal{K}}$ for the submatrix of the columns indexed by $\mathcal{K}$, and $D_{\mathcal{K}}$ for the diagonal submatrix of $D$ on the rows and columns indexed by $\mathcal{K}$.

\begin{lemma}\label{lem:IS}
$F_{\mathcal{I}}=S\,A_{\mathcal{I}}\,D_{\mathcal{I}}$. In particular every information set for $A$ is an information set for $F$, and conversely.
\end{lemma}

\begin{proof}
Because $D$ is diagonal, right multiplication by $D$ scales columns individually, so $(AD)_{\mathcal{I}}=A_{\mathcal{I}}D_{\mathcal{I}}$. Hence by \eqref{eq:AD}, $F_{\mathcal{I}}=(SAD)_{\mathcal{I}}=S(AD)_{\mathcal{I}} =SA_{\mathcal{I}}D_{\mathcal{I}}$. As $S$ and $D_{\mathcal{I}}$ are invertible, $F_{\mathcal{I}}$ is invertible if and only if $A_{\mathcal{I}}$ is. The converse statement follows by applying the same argument to $S^{-1}FD^{-1}=A$.
\end{proof}

We may assume that $\mathcal{I}=[k]$. Indeed, let $\sigma$ be a permutation matrix carrying $\mathcal{I}$ onto $[k]$; replacing $A$ by $A\sigma$ and $F$ by $F\sigma$ turns \eqref{eq:AD} into $F\sigma=S(A\sigma)(\sigma^{-1}D\sigma)$, and $\sigma^{-1}D\sigma$ is again an invertible diagonal matrix, obtained from $D$ by the permutation of its entries induced by $\sigma$. Since $\sigma$ is known, recovering $\sigma^{-1}D\sigma$ recovers $D$. We therefore fix $\mathcal{I}=[k]$ and put $\mathcal{J}:=[n]\setminus\mathcal{I}=\{k+1,\dots,n\}$, and we write
\[
D=\begin{pmatrix} D_{\mathcal{I}} & 0\\ 0 & D_{\mathcal{J}}\end{pmatrix}, \qquad D_{\mathcal{I}}=\diag(d_1,\dots,d_k), \qquad D_{\mathcal{J}}=\diag(d_{k+1},\dots,d_n).
\]

Putting the two matrices in systematic form, we have:
\[
A':=A_{\mathcal{I}}^{-1}A=(\,I_k \mid A''\,), \qquad F':=F_{\mathcal{I}}^{-1}F=(\,I_k \mid F''\,), \qquad S':=F_{\mathcal{I}}^{-1}SA_{\mathcal{I}},
\]
where $I_k$ denotes the $k\times k$ identity matrix and $A'',F''\in\F_q^{k\times(n-k)}$. Both systematic forms exist by Lemma \ref{lem:IS}, and $S'$ is invertible.

\begin{proposition}\label{prop:reduce}
With the notation above,
\begin{equation}\label{eq:diageq}
S'=D_{\mathcal{I}}^{-1} \qquad\text{and}\qquad D_{\mathcal{I}}^{-1}\,A''\,D_{\mathcal{J}}=F'' .
\end{equation}
Consequently $S=F_{\mathcal{I}}\,D_{\mathcal{I}}^{-1}\,A_{\mathcal{I}}^{-1}$.
\end{proposition}

\begin{proof}
We first check that $S'A'D=F'$:
\[
S'A'D =F_{\mathcal{I}}^{-1}SA_{\mathcal{I}}\cdot A_{\mathcal{I}}^{-1}A\cdot D =F_{\mathcal{I}}^{-1}SAD =F_{\mathcal{I}}^{-1}F =F' ,
\]
using \eqref{eq:AD} in the third step. Now, we expand both sides blockwise. Since $A'=(I_k\mid A'')$ and $D$ is block diagonal with respect to $\mathcal{I}\cup\mathcal{J}$,
\[
A'D=(\,I_kD_{\mathcal{I}}\mid A''D_{\mathcal{J}}\,) =(\,D_{\mathcal{I}}\mid A''D_{\mathcal{J}}\,), \qquad\text{so}\qquad S'A'D=(\,S'D_{\mathcal{I}}\mid S'A''D_{\mathcal{J}}\,).
\]
Comparing with $F'=(I_k\mid F'')$ block by block gives $S'D_{\mathcal{I}}=I_k$, whence $S'=D_{\mathcal{I}}^{-1}$, and $S'A''D_{\mathcal{J}}=F''$, which upon substituting $S'$ is the second identity of \eqref{eq:diageq}. Finally, from $S'=F_{\mathcal{I}}^{-1}SA_{\mathcal{I}}$ we get $S=F_{\mathcal{I}}S'A_{\mathcal{I}}^{-1} =F_{\mathcal{I}}D_{\mathcal{I}}^{-1}A_{\mathcal{I}}^{-1}$.
\end{proof}

Thus, we are only left with the task of finding $D$ itself. The second identity in \eqref{eq:diageq} is precisely an instance of the diagonal equivalence problem, formulated and solved efficiently in \cite[\S 4.4]{EngSch}.

\begin{definition}[Diagonal equivalence]\label{def:diageq}
Let $A,B\in\F^{r\times s}$. Then $A$ is {diagonally equivalent} to $B$ if there exist invertible diagonal matrices $X\in\F^{r\times r}$ and $Y\in\F^{s\times s}$ with $XAY^{-1}=B$.
\end{definition}

Taking $r=k$, $s=n-k$, $A=A''$, $B=F''$, $X=D_{\mathcal{I}}^{-1}$ and $Y=D_{\mathcal{J}}^{-1}$, the identity $D_{\mathcal{I}}^{-1}A''D_{\mathcal{J}}=F''$ reads $XA''Y^{-1}=F''$. Thus:

\begin{corollary}\label{cor:reduction}
The recovery of $D$ from $P$ reduces, in $O(nk^2)$ field operations, to solving one instance of the diagonal equivalence problem for a pair of matrices in $\F_q^{k\times(n-k)}$. Any solution $(X,Y)$ of that instance yields $D=\diag(X^{-1},Y^{-1})$ and then $S=F_{\mathcal{I}}\,X\,A_{\mathcal{I}}^{-1}$.
\end{corollary}

\subsection{The Engel--Schneider algorithm}

Engel and Schneider \cite{EngSch} solve Definition \ref{def:diageq} over any abelian group with zero. We give below a short summary of their algoirthm.

Consider the identity $D_{\mathcal{I}}^{-1}A''D_{\mathcal{J}}=F''$ entrywise. For $i\in\mathcal{I}$ and $j\in\mathcal{J}$ its $(i,j)$ entry is $d_i^{-1}A''_{ij}d_j$, so the identity is equivalent to the following system of constraints:
\begin{equation}\label{eq:ratios}
A''_{ij}\,\frac{d_j}{d_i}=F''_{ij} \qquad\qquad (i\in\mathcal{I},\ j\in\mathcal{J}).
\end{equation}
Note that a position with $A''_{ij}=0$ imposes $F''_{ij}=0$ as well; and at a position with $A''_{ij}\neq 0$, we have \emph{ratio}
\begin{equation}\label{eq:weight}
\frac{d_j}{d_i}=w_{ij}, \qquad w_{ij}:=\frac{F''_{ij}}{A''_{ij}}\in\F_q^{*} .
\end{equation}
We then build a bipartite graph $\mathcal{B}$ on the vertex set $\mathcal{I}\cup\mathcal{J}=[n]$ having an edge $\{i,j\}$ exactly when $A''_{ij}\neq 0$, each edge is labelled with the weight $w_{ij}$.

Next, we run a breadth-first search on $\mathcal{B}$ to obtain a spanning forest (a spanning tree for each connected components of the graph) together with one root $r_1,\dots,r_t$ for each of the $t$ connected components. BFS runs in time which is linear in the size of $\mathcal{B}$, and produces an ordering for each spanning tree.

We set $d_{r_p}:=1$ for each root. Then we sweep the vertices in BFS order. When we reach a vertex $v$ from an already assigned neighbour $u$, we set
\[
d_v:=w_{uv}\,d_u \quad \text{if } u\in\mathcal{I},\ v\in\mathcal{J}, \qquad\qquad d_v:=d_u/w_{vu} \quad \text{if } u\in\mathcal{J},\ v\in\mathcal{I},
\]
according to \eqref{eq:weight}.

We set $D:=\diag(d_1,\dots,d_n)$ and then calculate $S=F_{\mathcal{I}}D_{\mathcal{I}}^{-1}A_{\mathcal{I}}^{-1}$ by Proposition \ref{prop:reduce}. Correctness is guaranteed by \cite[Thm.~3.8, Cor.~3.11]{EngSch}.

The Engel-Schneider process for solving diagonal equivalence is summarised in Algorithm \ref{alg:diag}, and the complete reduction, with all pre and post processing steps is summarised in Algorithm \ref{alg:full}.

\begin{algorithm}[htbp]
\caption{\textsc{DiagonalEquivalence}}
\label{alg:diag}
\begin{algorithmic}[1]
\Require $A'',F''\in\F_q^{k\times(n-k)}$, with rows indexed by $\mathcal{I}=[k]$ and columns by $\mathcal{J}=\{k+1,\dots,n\}$.
\Ensure $(d_1,\dots,d_n)\in(\F_q^{*})^{n}$ satisfying \eqref{eq:ratios}, i.e.\ $D_{\mathcal{I}}^{-1}A''D_{\mathcal{J}}=F''$ for $D=\diag(d_1,\dots,d_n)$; or \NO{} if no such tuple exists.
\Statex
\State \textbf{if} $\exists (i,j)$ for which $A''_{ij}=0\not\Leftrightarrow F''_{ij}=0$ \textbf{then return} \NO{} 
\State $E\gets\big\{\{i,j\}\ :\ i\in\mathcal{I},\ j\in\mathcal{J},\ A''_{ij}\neq 0\big\}$, \quad $\mathcal{B}\gets\big([n],E\big)$
\State $w_{ij}\gets F''_{ij}\,/\,A''_{ij}$ \quad for each $\{i,j\}\in E$
\State run a breadth-first search on $\mathcal{B}$, obtaining a spanning forest, its roots $r_1,\dots,r_t$ (one per connected component) and, for each non-root $v$, the vertex $u(v)$ from which $v$ was first reached
\State $d_{r_p}\gets 1$ \quad for $p=1,\dots,t$
\For{$v\in[n]\setminus\{r_1,\dots,r_t\}$ in breadth-first order}
  \State $u\gets u(v)$
  \State $d_v\gets\begin{cases}
           w_{uv}\cdot d_u, & u\in\mathcal{I}\ \ (\text{so }v\in\mathcal{J}),\\[2pt]
           d_u\,/\,w_{vu}, & u\in\mathcal{J}\ \ (\text{so }v\in\mathcal{I}).
         \end{cases}$
\EndFor
\State \textbf{if} $\exists (i,j) \in E$ for which $A''_{ij}\,d_j/d_i\neq F''_{ij}$ \textbf{then return} \NO{} \Comment{non-forest edges can fail}
\State \Return $(d_1,\dots,d_n)$
\end{algorithmic}
\end{algorithm}

\begin{algorithm}[htbp]
\caption{\textsc{SearchLCE}}
\label{alg:full}
\begin{algorithmic}[1]
\Require Linearly equivalent $G,F\in\F_q^{k\times n}$ of full row rank $k$; an oracle $\mathcal{O}_\mathrm{LCE}$ for LCE.
\Ensure $S\in\GL_k(\F_q)$, a permutation matrix $P$ and an invertible diagonal $D$ with $SG\,PD=F$; or \NO.
\Statex
\Statex \textit{Preprocessing}
\State $Z_G\gets\{i\in[n]:g_i=0\}$, \quad $Z_F\gets\{j\in[n]:f_j=0\}$
\State $\bar n\gets n-|Z_G|$; \ let $\iota_G:[\bar n]\to[n]\setminus Z_G$ and
       $\iota_F:[\bar n]\to[n]\setminus Z_F$ be the increasing bijections
\State $\bar G\gets$ the columns of $G$ indexed by $[n]\setminus Z_G$, \quad
       $\bar F\gets$ the columns of $F$ indexed by $[n]\setminus Z_F$

\Statex
\Statex \textit{Recovering the permutation.}
\State $\bar\pi\gets\Call{RecoverPermutation}{\bar G,\bar F,\mathcal{O}_\mathrm{LCE}}$
       
\State $\bar P\gets$ the permutation matrix of $\bar\pi$
\Statex
\Statex \textit{Systematic formulation}
\State $A\gets\bar G\bar P$
      
\State $A'\gets\rref(A)$, \quad $\mathcal{I}\gets\piv(A')$
       
\State $F'\gets\rref(\bar F)$
       
\State $\sigma\gets$ the permutation matrix of the permutation of $[\bar n]$
       carrying $\mathcal{I}$ onto $[k]$ and $[\bar n]\setminus\mathcal{I}$ onto
       $\{k+1,\dots,\bar n\}$, both increasingly
\State $A\gets A\sigma$, \ $\bar F\gets\bar F\sigma$, \
       $\big(I_k\mid A''\big)\gets A'\sigma$, \
       $\big(I_k\mid F''\big)\gets F'\sigma$
\State $\mathcal{I}\gets[k]$, \quad $\mathcal{J}\gets\{k+1,\dots,\bar n\}$
       \Comment{now $\bar F=SA\widetilde D$,
                $\widetilde D=\sigma^{-1}\bar D\sigma$}
\Statex
\Statex \textit{Recovering the diagonal and change of basis}
\State $(d_1,\dots,d_{\bar n})\gets\Call{DiagonalEquivalence}{A'',F''}$
       
\State $\widetilde D\gets\diag(d_1,\dots,d_{\bar n})$
\State $S\gets\bar F_{\mathcal{I}}\,\widetilde D_{\mathcal{I}}^{-1}\,
       A_{\mathcal{I}}^{-1}$
      
\State $\bar D\gets\sigma\,\widetilde D\,\sigma^{-1}$
       
\Statex
\Statex \textit{Postprocessing}
\State $\beta\gets$ any bijection $Z_G\to Z_F$
\State $\pi\gets\big(\iota_F\circ\bar\pi\circ\iota_G^{-1}\big)\cup\beta$;
       \quad $P\gets$ the permutation matrix of $\pi$
\State $D\gets\diag(d^{*}_1,\dots,d^{*}_n)$, \quad where
       $d^{*}_j=\bar D_{\iota_F^{-1}(j)\,\iota_F^{-1}(j)}$ for $j\notin Z_F$ and
       $d^{*}_j=1$ for $j\in Z_F$
\State \Return $(S,P,D)$
\end{algorithmic}
\end{algorithm}

\subsection{Cost}

All costs below are counted in field operations in $\F_q$, each of which is $\mathrm{polylog}(q)$ bit operations. By the count of \cite[\S 5, Table 1]{EngSch} they require
\[
6n-2t \ \text{ storage operations}, \qquad n-t \ \text{ multiplications or divisions},
\]
and a total number of logical operations bounded above by $2n+t+n^{2}\le 3n+n^{2}$, as $t\le n$. Other costs include $O(nk)$ to form $A=GP$, $O(nk^{2})$ to find $\mathcal{I}$ and compute the systematic forms $A'$ and $F'$, and $O(k^{\omega})$ to recover $S$ from $D$, for a total of $O(nk^{2}+k^{\omega})$ field operations. In particular the recovery of $D$ is negligible compared to the $O(n^{2})$ oracle queries of the previous section, and the whole of it is deterministic and polynomial in $n$, $k$ and $\log q$.

\begin{gce}
\section{The general code equivalence problem}\label{sec:general}

We now prove Theorem \ref{thm:general}. Note that nothing in Section \ref{sec:main} needs to be redone. This is because of Lemma \ref{lem:frobclass}, which says a field automorphism leaves the partition into projective classes untouched. Thus, we can apply the first step to the pair of matrices $\big(\alpha(G),F\big)$ for an $\alpha$ for which an equivalence exists. 

Throughout, $G,F\in\F_q^{k\times n}$ have full row rank $k$ and no zero column, and we assume, having spent one oracle call, that $(G,F)$ is a \YES{} instance of GCE. By Lemma \ref{lem:frobclass} and Lemma \ref{lem:invariance} applied to $\big(\alpha(G),F\big)$ we may still set $m:=m(G)=m(F)$, as this value does not depend on $\alpha$.

\subsection{Recovering the permutation}

The induction hypothesis is now as follows:

\medskip
\noindent{There exist $S_{T-1}\in\GL_k(\F_q)$, $M_{T-1}\in\Mon_n$ and $\alpha_{T-1}\in\Aut(\F_q)$ such that}
\[
S_{T-1}\,\alpha_{T-1}(G)\,M_{T-1}=F,\qquad \pi_{M_{T-1}}(t)=i_t\ \text{ for all } t\in A_{T-1}.
\]

\medskip

We note that $G''_T$ and $F''_T$ depend on the classes $K_j\in\K(G)$, the value of $m$, and the classes $L_j=$ class of $f_{i_{t_j}}$ in $\K(F)$, and these are all independent of $\alpha$. Hence the target $\widehat L_1,\dots,\widehat L_\ell,\widehat L$ appearing in Proposition \ref{prop:step} are the same regardless, and Algorithm \ref{alg:BQ} may be run on $(G,F)$ unchanged. For completeness, we prove the following.

\begin{proposition}\label{prop:stepgen}
Assume the induction hypothesis. Let $t\notin A_{T-1}$ have unrepresented class, and let $G''_T,F''_T$ be as constructed before Proposition \ref{prop:step}. The answer of the GCE oracle on $\big(G''_T,F''_T\big)$ decides whether there exist $S_T$, $M_T$ and $\alpha_T$ with $S_T\,\alpha_T(G)\,M_T=F$, $\pi_{M_T}|_{A_{T-1}}=\pi_{M_{T-1}}|_{A_{T-1}}$ and $\pi_{M_T}(t)=i_t$.
\end{proposition}

\begin{proof}
Suppose first that such a triple exists. By Lemma \ref{lem:frobclass}, $\alpha_T(G''_T)=\alpha_T(G)''_T$. We apply the same construction of $M''_T$ as in item (i) of the proof of Proposition \ref{prop:step}, carried out for the pair $\big(\alpha_T(G),F\big)$. This gives us $S_T\,\alpha_T(G''_T)\,M''_T=F''_T$, so the oracle answers positively.

For the other direction, if the oracle answers positively, let $S''$, $M''$ and $\alpha''$ be such that $S''\,\alpha''(G''_T)\,M''=F''_T$. Since we have $\K\big(\alpha''(G''_T)\big)=\K\big(G''_T\big)$ by Lemma \ref{lem:frobclass}, every class of $\alpha''(G''_T)$ has the same size as the corresponding class of $G''_T$. The pair $\big(\alpha''(G''_T),F''_T\big)$ satisfies a linear equivalence, so Lemma \ref{lem:normalise} can be applied to it, and so too can the remainder of the proof of Proposition \ref{prop:step} with $G''_T$ replaced by $\alpha''(G''_T)$ and $G$ by $\alpha''(G)$. It produces $M_T$ with $S''\,\alpha''(G)\,M_T=F$, agreeing with $\pi_{M_{T-1}}$ on $A_{T-1}$ and sending $t$ to $i_t$. Setting $S_T:=S''$ and $\alpha_T:=\alpha''$ concludes the proof.
\end{proof}

Note that for a represented class, which uses Lemma \ref{lem:normalise} on $\big(\alpha_{T-1}(G),F\big)$, we use no oracle call. Thus, we have the following.

\begin{corollary}\label{cor:permgen}
Algorithm \ref{alg:perm}, run with a GCE oracle in place of $\mathcal{O}_\mathrm{LCE}$, terminates and returns $\pi\in\Sn$ for which there exist $S\in\GL_k(\F_q)$, $M\in\Mon_n$ and $\alpha\in\Aut(\F_q)$ with $S\,\alpha(G)\,M=F$ and $\pi_M=\pi$. It performs at most $n^2$ oracle calls, on instances of dimension $k$ and length at most $n+\frac{2}{27}n^3+O(n^2)$.
\end{corollary}

\subsection{Recovering the automorphism}

Let $\pi$ be the output of Corollary \ref{cor:permgen} and $P$ its permutation matrix. For each of the $\log_p q$ automorphisms $\alpha\in\Aut(\F_q)$, we run the process of Section \ref{sec:diagonal} on the pair $\big(\alpha(G),F\big)$. For the correct $\alpha$, this returns $S$ and $D$ with $S\,\alpha(G)\,PD=F$. For an incorrect one the pair is a \NO{} instance of LCE, and we can check for that using the following:
\begin{enumerate}
\item $\mathcal{I}=\piv\big(\rref(\alpha(G)P)\big)$ need not be an information set for $F$, thus we reject $\alpha$ if $\det F_{\mathcal{I}}=0$.
\item We may not have a case of diagonal equivalence, in which case there are two tests inside Algorithm \ref{alg:diag}. First, we check the supports of $A''$ and $F''$. Then we check that the ratios \eqref{eq:weight} are consistent for all the non-forest edges, the forest edges holding by construction. If either of these checks fail, Algorithm \ref{alg:diag} returns \NO{} and $\alpha$ is rejected.
\end{enumerate}
Note that these checks are both necessary and sufficient. Indeed if they have the same information set, then we can calculate $F_{\mathcal{I}}^{-1}F$, and we can apply Proposition \ref{prop:reduce}. If the second check passes, then the tuple returned lies in $(\F_q^{*})^n$ and satisfies $D_{\mathcal{I}}^{-1}A''D_{\mathcal{J}} =F''$, which by Proposition \ref{prop:reduce} is equivalent to $S\,\alpha(G)\,PD=F$ for $S:=F_{\mathcal{I}}D_{\mathcal{I}}^{-1} A_{\mathcal{I}}^{-1}$. The procedure is summarised in Algorithm \ref{alg:gce}.

\begin{algorithm}[htbp]
\caption{\textsc{SearchGCE}}
\label{alg:gce}
\begin{algorithmic}[1]
\Require $G,F\in\F_q^{k\times n}$ of full row rank $k$; an oracle $\mathcal{O}_{\mathrm{GCE}}$ for GCE.
\Ensure $S\in\GL_k(\F_q)$, a permutation matrix $P$, an invertible diagonal $D$ and $\alpha\in\Aut(\F_q)$ with $S\,\alpha(G)\,PD=F$; or \NO.
\Statex
\Statex \textit{Preprocessing}
\State $Z_G\gets\{i\in[n]:g_i=0\}$, \quad $Z_F\gets\{j\in[n]:f_j=0\}$
\State $\bar n\gets n-|Z_G|$; \ let $\iota_G:[\bar n]\to[n]\setminus Z_G$ and $\iota_F:[\bar n]\to[n]\setminus Z_F$ be the increasing bijections
\State $\bar G\gets$ the columns of $G$ indexed by $[n]\setminus Z_G$, \quad $\bar F\gets$ the columns of $F$ indexed by $[n]\setminus Z_F$
\If{$\mathcal{O}_{\mathrm{GCE}}(\bar G,\bar F)=\NO$} \Return \NO \EndIf
\Statex
\Statex \textit{Recovering the permutation.}
\State $\bar\pi\gets\Call{RecoverPermutation}{\bar G,\bar F, \mathcal{O}_{\mathrm{GCE}}}$
\State $\bar P\gets$ the permutation matrix of $\bar\pi$
\Statex
\Statex \textit{Recovering the automorphism, diagonal, and change of basis}
\For{$s=0,1,\dots,e-1$ \textbf{where} $q=p^{e}$}
  \State $\alpha\gets\big(x\mapsto x^{p^{s}}\big)$
  \Statex
  \Statex \quad\ \textit{Systematic formulation}
  \State $A\gets\alpha(\bar G)\,\bar P$
  \State $A'\gets\rref(A)$, \quad $\mathcal{I}\gets\piv(A')$
  \If{$\det\bar F_{\mathcal{I}}=0$} \textbf{continue} \EndIf 
  \State $F'\gets\bar F_{\mathcal{I}}^{-1}\bar F$
  \State $\sigma\gets$ the permutation matrix of the permutation of $[\bar n]$ carrying $\mathcal{I}$ onto $[k]$ and $[\bar n]\setminus\mathcal{I}$ onto $\{k+1,\dots,\bar n\}$, both increasingly
  \State $A\gets A\sigma$, \ $\bar F\gets\bar F\sigma$, \ $\big(I_k\mid A''\big)\gets A'\sigma$, \ $\big(I_k\mid F''\big)\gets F'\sigma$
  \State $\mathcal{I}\gets[k]$, \quad $\mathcal{J}\gets\{k+1,\dots,\bar n\}$ \Comment{now $\bar F=S A \widetilde D$, $\widetilde D=\sigma^{-1}\bar D\sigma$}
  \Statex
  \Statex \quad\ \textit{Recovering the diagonal and change of basis}
  \State $\mathbf{d}\gets\Call{DiagonalEquivalence}{A'',F''}$
  \If{$\mathbf{d}=$\NO{}} \textbf{continue} \EndIf 
  \State $\widetilde D\gets\diag(\mathbf{d})$
  \State $S\gets\bar F_{\mathcal{I}}\,\widetilde D_{\mathcal{I}}^{-1}\, A_{\mathcal{I}}^{-1}$
  \State $\bar D\gets\sigma\,\widetilde D\,\sigma^{-1}$
  \Statex
    \Statex \qquad \textit{Postprocessing}
    \State $\beta\gets$ any bijection $Z_G\to Z_F$
    \State $\pi\gets\big(\iota_F\circ\bar\pi\circ\iota_G^{-1}\big)\cup\beta$; \quad $P\gets$ the permutation matrix of $\pi$
    \State $D\gets\diag(d^{*}_1,\dots,d^{*}_n)$, \quad where $d^{*}_j=\bar D_{\iota_F^{-1}(j)\,\iota_F^{-1}(j)}$ for $j\notin Z_F$ and $d^{*}_j=1$ for $j\in Z_F$
    \State \Return $(S,P,D,\alpha)$
\EndFor
\State \Return \NO
\end{algorithmic}
\end{algorithm}

\subsection{Cost}

Let $q=p^{e}$, so that $e=\log_p q$. Algorithm \ref{alg:diag} is executed at most $e$ times, each iteration costing $O(nk^2+k^{\omega})$ field operations by the count of Section \ref{sec:diagonal}. For computing $\alpha(\bar G)$, it costs $O(e^2)$ operations in $\F_p$ per entry, that is $O(nke^2)$ operations in $\F_p$ in total. Every operation in $\F_q$ is $\mathrm{polylog}(q)$ bit operations. Hence, the whole reduction runs in time polynomial in $n$, $k$ and $\log q$.
\end{gce}

\section{Conclusion and Future Work}\label{sec:conclusion}

We have presented efficient search-to-decision reductions for the Linear and Generalised Code Equivalence problems. As a direct result of our work, we now have a constructive method for converting any decision LCE (or GCE) solver into a search solver in polynomial time.

What we don't show is a reduction from search LCE to GCE (or search GCE). Indeed, the naive strategy of enumerating $\alpha$ and asking if $\big(\alpha(G),F\big)$ is an instance of LCE gives us a reduction from search or decision GCE to search or decision LCE.

Finally, while we now have settled the search-to-decision reductions for the Hamming metric CE variants (PCE, LCE, and GCE), an analogous search-to-decision reduction for Matrix Code Equivalence (MCE), which forms the basis of MEDS \cite{MEDS}, is still missing. These remain interesting open problems for future research.

\bibliographystyle{plain} 
\bibliography{references}

@INPROCEEDINGS{BM23,
  author={Biasse, Jean-François and Micheli, Giacomo},
  booktitle={2023 IEEE International Symposium on Information Theory (ISIT)}, 
  title={A Search-to-Decision Reduction for the Permutation Code Equivalence Problem}, 
  year={2023},
  volume={},
  number={},
  pages={602-607},
  note={\url{https://doi.org/10.1109/ISIT54713.2023.10206940}}}

@incollection{LESS20,
  author       = {Biasse, Jean-Fran\c{c}ois and Micheli, Giacomo and Persichetti, Edoardo and Santini, Paolo},
  title        = {{LESS} is more: code-based signatures without syndromes},
  booktitle    = {Progress in Cryptology -- AFRICACRYPT 2020},
  series       = {LNCS},
  volume       = {12174},
  publisher    = {Springer},
  year         = {2020},
  pages        = {45--65},
  note = {\url{https://doi.org/10.1007/978-3-030-51938-4_3}}
}

@ARTICLE{BMS25,
  author={Battagliola, Michele and Mora, Rocco and Santini, Paolo},
  journal={IEEE Transactions on Information Theory}, 
  title={Using the Schur Product to Solve the Code Equivalence Problem}, 
  year={2026},
  volume={72},
  number={7},
  pages={4675-4694},
  note={\url{https://doi.org/10.1109/TIT.2026.3694630}}}

@article{CPS25,
  author       = {Chou, Tung and Persichetti, Edoardo and Santini, Paolo},
  title        = {On linear equivalence, canonical forms, and digital signatures},
  journal      = {Designs, Codes and Cryptography},
  volume       = {93},
  number       = {7},
  pages        = {2415--2457},
  year         = {2025},
  note = {\url{https://doi.org/10.1007/s10623-025-01576-1}},
 
}

@Article{DMS26,
author={D'Alconzo, Giuseppe
and Meneghetti, Alessio
and Signorini, Edoardo},
title={Group factorisation for smaller signatures from cryptographic group actions},
journal={Designs, Codes and Cryptography},
year={2026},
month={Feb},
day={02},
volume={94},
number={2},
pages={42},
issn={1573-7586},
note={\url{https://doi.org/10.1007/s10623-025-01787-6}}
}

@article{EngSch,
  author       = {Engel, G. M. and Schneider, H.},
  title        = {Algorithms for testing the diagonal similarity of matrices and related problems},
  journal      = {SIAM Journal on Algebraic Discrete Methods},
  volume       = {3},
  number       = {4},
  pages        = {429--438},
  year         = {1982},
  note = {\url{https://doi.org/10.1137/0603044}}
}

@misc{AD26,
  author       = {Alecci, M. and D'Alconzo, G.},
  title        = {Linear code equivalence via Pl\"ucker coordinates},
  year         = {2026},
  note = {arXiv:2603.09869 \url{https://doi.org/10.48550/arXiv.2603.09869}}
}

@inproceedings{PS23,
author = {Persichetti, Edoardo and Santini, Paolo},
title = {A New Formulation of the Linear Equivalence Problem and Shorter LESS Signatures},
year = {2023},
isbn = {978-981-99-8738-2},
publisher = {Springer-Verlag},
address = {Berlin, Heidelberg},
note = {\url{https://doi.org/10.1007/978-981-99-8739-9_12}},
doi = {10.1007/978-981-99-8739-9_12},
booktitle = {Advances in Cryptology – ASIACRYPT 2023: 29th International Conference on the Theory and Application of Cryptology and Information Security, Guangzhou, China, December 4–8, 2023, Proceedings, Part VII},
pages = {351–378},
numpages = {28},
location = {Guangzhou, China}
}

@InProceedings{MEDS,
author="Chou, Tung
and Niederhagen, Ruben
and Persichetti, Edoardo
and Randrianarisoa, Tovohery Hajatiana
and Reijnders, Krijn
and Samardjiska, Simona
and Trimoska, Monika",
editor="El Mrabet, Nadia
and De Feo, Luca
and Duquesne, Sylvain",
title="Take Your MEDS: Digital Signatures from Matrix Code Equivalence",
booktitle="Progress in Cryptology - AFRICACRYPT 2023",
year="2023",
publisher="Springer Nature Switzerland",
address="Cham",
pages="28--52",
isbn="978-3-031-37679-5",
note = {\url{https://doi.org/10.1007/978-3-031-37679-5_2}}
}

@InProceedings{survey24,
author="Horlemann, Anna-Lena
and Mazumder, Abhinaba
and Schaller, Michael
and Weger, Violetta",
editor="Batina, Lejla
and {\"O}zbudak, Ferruh",
title="A Survey on Code Equivalence: The State-of-the-Art and Open Questions",
booktitle="Arithmetic of Finite Fields",
year="2026",
publisher="Springer Nature Switzerland",
note={\url{https://doi.org/10.1007/978-3-032-27574-5_2}},
address="Cham",
pages="12--31",
isbn="978-3-032-27574-5"
}

\end{document}